\documentclass[runningheads]{llncs}
\usepackage[T1]{fontenc}

\usepackage{amssymb, mathtools, thm-restate}
\usepackage[bookmarks, bookmarksnumbered, linktocpage, urlcolor = cyan]{hyperref}
\usepackage{cleveref}

\newcommand{\E}{\mathbb{E}}
\newcommand{\Prob}{\mathbb{P}}

\begin{document}
\title{A contribution to the critique of blockchain censorship}

\author{Ruichao Jiang\inst{1} \and Michelle Yeo\inst{2,3} \and Long Wen\inst{1}}

\institute{Derivation Technology Ltd, Hong Kong, China  \email{\{ruichao.jiang,long.wen\}@derivation.info}\\ \and Aarhus University, Aarhus, Denmark\\ \email{mxyeo@cs.au.dk}
\and Nanyang Technological University }
\maketitle

\begin{abstract}
    We study the blockchain censorship attack introduced in \cite{yeo}, which shows that joining the attack is a dominant strategy. We show that, by introducing certain detectability threshold, joining the attack can lead to strictly less reward for whales, which are defined to be a small number of validators that hold significantly more voting power than the rest (henceforth known as minnows). This leads to a change of the equilibrium: With whales unwilling to participate in the attack, it is difficult for minnows alone to launch the attack. We also perform Monte Carlo simulation to show the existence of reduction for whales' reward in Ethereum and Solana.
\end{abstract}

\section{Introduction}
Blockchain censorship is a growing topic of interest and has been studied extensively on both the transaction-level (see~\cite{abraham,alpos,baum,cachin,kostiainen}),
and the consensus-level (see~\cite{bonneau,mcmorry}).
Recent work of Yeo and Zhang~\cite{yeo} analyzes rational censorship attacks in a blockchain model where validators arrive in random order and coordinate their attack intention via only a single message through a shared ``blackboard'' (i.e., the blockchain). 
They show that under some assumptions on attacking coalition size and order of validators joining the attack, both attacking the blockchain as well as joining the censorship  attack gives validators larger profits (formally, the attack is a subgame perfect equilibrium in the underlying game induced by the attack).
In so doing, they show how strategic interactions between validators can negatively impact the liveness of the underlying blockchain. 

The attack presented in~\cite{yeo} lies on four key assumptions.
Firstly, all players are rational, which means that as long as excluding some other validator can result in larger expected profits shared among the non-excluded validators, the non-excluded validators have an incentive to do so.
Secondly, the attacking coalition needs to have sufficient cumulative voting power, in order to prevent the excluded validators from taking over and excluding the attacking coalition.
Thirdly, the ``no big player'' assumption, which states that there is no player that controls too large amount of voting power such that the player has to be included in all attacking coalitions. 
Lastly, the attack also assumes the existence of a detectability threshold, which is an exclusion threshold $\eta$ such that excluding more than $\eta$ users will lead to the attack being detected and hence a larger cost (represented by the drop in the price of the underlying cryptocurrency) for the attacking validators relative to their gain in profits from the attack.

In this work, we reexamine the attack of~\cite{yeo} from a complementary angle. In particular, we note that under a specific and realistic distribution of voting power, there are parameter regimes over the attacking coalition size and total number of validators such that not all validators are guaranteed profits in expectation. 
In more detail, we model voting power as distributed between two types of validators: whales and minnows. 
Whale validators control noticeably larger voting power compared to minnows (see~\Cref{sec:whale-minnow} for more details), and it is these whale validators that suffer a loss in expected profits in joining the attack compared to not under certain parameter regimes. 
Consequently, this implies parameter regimes under which the original attack in~\cite{yeo} might not succeed, due to the whales' reluctance to join the attack (and insufficient cumulative voting power from the minnows).

We stress that our work does not contradict the result in~\cite{yeo}. 
Rather, our analysis relies on a subtle insight: although the model proposed by~\cite{yeo} imposes a detectability cost, their analysis assumes that detection is ``in-built'' into the attack and thus the attack either stops before being detected or remains undetected. This effectively guarantees only profits for the users who are in the attacking coalition. 
In our work, we allow for cases where the attack can be detected, and hence the detectability cost can actually be imposed.



\subsection{Our contributions}
Our contributions can be summarised as follows:

\begin{itemize}
    \item (Model.) We examine the original censorship attack proposed by~\cite{yeo} under a realistic distribution of validator voting power, where validators are split into whales (nodes that control significantly larger voting powers) and minnows (nodes that control uniformly small voting powers).
    \item (Analysis.) We analyse parameter regimes whereby whale nodes gain a lower expected profit compared to their expected profit under the honest strategy, due to the detection cost.
    \item (Numerical experiments.) We validate our theoretical bounds with Monte Carlo simulations using real-world stake distribution data from Ethereum and Solana.
\end{itemize}


\subsection{Related work}

\paragraph{Blockchain censorship attacks.}
Blockchain censorship is a topic of growing interest, especially in the light of recent top-down censorship mandated by governments to target malicious individuals and organisations (see e.g.,~\cite{ofac,us-blender,us-tornado-case}).
On the research front, blockchain censorship has mainly been studied in the context of miner extractable value (MEV) prevention (see, e.g.,~\cite{abraham,alpos,baum,cachin,kostiainen}).
A comprehensive survey of blockchain censorship can be found in~\cite{wahrstatter}. 
The work which is most related to ours is the censorship attack presented by~\cite{yeo}. Their attack is based on the simple observation that excluding some number of validators (i.e., ignoring their messages and blocks) would only benefit the non-excluded validators in that they will achieve a larger proportional share of block rewards. Additionally,~\cite{yeo} also assumes that too large an attack might be detectable and thus detrimental to the blockchain system. They model this by imposing a detection cost (i.e., a fixed cost representing the drop of the underlying cryptocurrency price) if too many validators are excluded. From these observations, they then show that (1) every validator has an incentive to propose forming a coalition that excludes some amount of validators, and (2) assuming validators join the coalition in a randomised order, every validator has an incentive to join the attacking coalition. Our work takes the attack of~\cite{yeo} as a starting point, but analyses the attack given a specific distribution of voting powers among validator nodes: validators are classified into whales and minnows with whales holding on to much more power compared to minnows. Our work shows that under this specific distribution of power, the detection cost might actually deter whales from joining the attack, which adds nuance to the original attack by~\cite{yeo}.

\paragraph{Detectability.} 
The threat of attack detectability and its use in disincentivising attacks was first introduced by~\cite{aumann} with their notion of a covert adversary: an adversary that wants to minimise attack detection. 
Follow up work in the blockchain setting also examined detectability costs in the context of selfish mining (see~\cite{bahrani}), cryptographic self-selection (see~\cite{cai}), and blockchain censorship (see~\cite{yeo}).
Notably, the work of~\cite{yeo} did not examine the impact of the detectability cost on their proposed attack.
Our work goes one important step further with the detectability cost proposed by~\cite{yeo}: we argue that assuming the cost is prohibitively large, there are attack settings where whales might not want to join the attack, which provides security against the ``costless'' censorship attacks as outlined in~\cite{yeo}.

\paragraph{Empirical analyses of stake distribution.}
Our work is also related to works that empirically analyse the distribution of power among validators in blockchains.
~\cite{grandjean} measure the impact of switching over from proof of work to proof of stake (PoS) in the decentralisation of Ethereum, and show that PoS Ethereum displays a huge concentration of power among a small number of validators. 
This property is also observed in~\cite{nir}, as well as a broader blockchain power centralisation literature review performed by~\cite{sai}. Altogether, this validates our model and study of the censorship attack using whales and minnows, as well as reinforces the results of our empirical analysis.

\subsection{Organisation}
We present notation and our model in~\Cref{sec:background}. In~\Cref{sec:whale-minnow}, we define our whale and minnow model and analyse the censorship attack under this specific whale and minnow power distribution. We then present our main theorem (~\Cref{thm:validity-range}) in~\Cref{sec:validity}, which shows precise conditions in which whale nodes would achieve a reduced profit from joining the attack compared to acting honestly. We supplement and justify our theoretical bounds with Monte Carlo simulations on real-world Ethereum and Solana data in~\Cref{sec:experiment}, and finally discuss open problems in~\Cref{sec:discussion}.

\section{Background and notation} \label{sec:background}

\paragraph{Notation.}
We use $[n]$ to denote the set $\{1, \dots, n\}$.
Given a set $S$, we use $\mathbf{1}\{S\}$ to denote the indicator function of $S$. For two functions $f$ and $g$, we follow the convention in probability theory to use $f\vee g$ to denote the pointwise maximum of two, i.e. $(f\vee g)(x)=\max\{f(x),g(x)\}$. If $g$ happens to be a number, it means a constant function.

\subsection{Glossary}\label{sec:Glossary}
For readability, we present a glossary of common notation used throughout this work.

\begin{description}
\item[$t$:] The requisite power threshold necessary to launch the censorship attack, which is described in~\Cref{sec:attackmodel}.
\item[$k$:] The attacking coalition size, i.e., number of parties.
\item[$W$ (resp. $M$):] Set of whale (resp. minnow) validator nodes.
\item[$\alpha$:] Voting power of each whale.
\item[$w$:] Number of whales.
\item[$S_k$:] Total voting power of an attacking coalition of size $k$ (\Cref{def:prefix-stake}).
\item[$K(t,k)$:] The actual number of participant when both the voting power and coalition size requirements are met.
\item[$N_k$:] Number of whales in an attacking coalition of size $k$ (\Cref{def:attacking-whales}).
\item[$N_0$:] Minimum number of whales in a successful attack of coalition size $k$ (\Cref{def:n0}).
\item[$\Phi(k,t)$:] Expected whale's gain for the whales that joined a successful attack (\Cref{def:whale's-gain}).
\item[$\overline{\Phi}(k,t)$:] An upper bound on $\Phi(k,t)$ (\Cref{proposition:tail-bound}).
\item[$m$:] Upper bound on $N_0$, the minimum number of whales in a successful attacking coalition.
\item[$\mathcal{F}$:] The feasible set of $m$.

\end{description}
Note that  $N_0$ and $m$ are just numbers but $S_k$, $K(t,k)$ and $N_k$ are random variables, whose sample space is the set of permutations. As such, they are functions of the arrival order $\pi$ (\Cref{def:pi}). We suppress this explicit dependence.
\subsection{Model}\label{sec:model}

\paragraph{Blockchain, system and threat model.}
We assume there are $n$ validators, indexed from $1$ to $n$. 
We assume an underlying blockchain system $S$ that is maintained by the aforementioned $[n]$ validators.
Each validator $x_i$ has some amount of voting power $v_i$ that determines the probability that they create the next block.
Wlog, we normalise all voting powers to sum to $1$, i.e., $\sum_{i=1}^n v_i =1$.
Each time a block is appended to the blockchain, the block reward is paid in full to the validator that proposed the block. 
Finally, we assume the underlying blockchain defines a protocol that
validators should follow, henceforth known as the \emph{honest strategy}. 

We assume time proceeds in discrete time steps where at each
time step a block is added to the blockchain. For simplicity purposes, we also
assume that changes in the state of the blockchain and financial system as a result of user actions at some time step will be immediately
reflected in the next time step, i.e., the consequences of user actions do
not incur any time lag.

We assume all validators are rational, i.e., each validator strives to maximise their utility. 
Moreover, we assume that each validator has the same utility function, which is to maximise their long-run average block reward.

\paragraph{Cost model.}
We assume the same cost model as in~\cite{yeo}. Namely, there is some cost if the censorship attack is detected.
Attack detectability is determined by the size of the attacking coalition, i.e., the attack is detectable if the attacking coalition size exceeds some detectability threshold $1-\eta$, and undetectable otherwise.
We assume the cost of detection is considerably larger than the expected rewards of launching the attack, thus validations halt the attack whenever the attacking coalition size exceeds the threshold (more details in the attack model below.).

\subsubsection{Attack model}\label{sec:attackmodel}
We summarise the attack model of~\cite{yeo}. In the attack model of~\cite{yeo}, the attack is initiated by posting a message to a smart contract on the blockchain signifying the intent to form an attacking coalition to censor messages from other users. Anyone can initiate this smart contract and thus the arrival order is uniform (We leave out the possibility that the validators can re-order the arrival order, which corresponds to a form of MEV). Following this initial message, the authors show if there is no requirement on the coalition size, the best response will be to always join the attack: If you don't join, you are sure to be censored. In our work, we model the uniformly random process in which validators join the attack by the following.

\begin{definition}[Arrival order] \label{def:pi}
An \emph{arrival order} is a permutation from the set of validators to the set of arrival orders
    \begin{equation*}
        \pi:\{1,\dots,n\}\to\{1,\dots,n\}.
    \end{equation*}
\end{definition}
Note that $\pi^{-1}(i)$ means the index of the validator that arrives in the $i$-th position. We therefore assume that we sample $\pi$ from a uniform distribution over all $n!$ permutations.

We also introduce the notion of $k$-sum, which tracks the total voting power accumulated by the first $k$ attacking validators. 

\begin{definition}[$k$-sum] \label{def:prefix-stake}
After $k$ arrivals, the voting power that attackers have accumulated is
    \begin{equation*}
        S_k(\pi)\coloneqq\sum_{i=1}^kv_{\pi^{-1}(i)}.
    \end{equation*}
    The dependence on $\pi$ can be omitted if no confusion is caused, especially when the subscript $k$ itself depends on the randomness of $\pi$.
\end{definition}

We now consider the stopping times of the random validator arrival process. Recall that in~\cite{yeo}, the random validator arrival process successfully halts when either one of the following conditions hold: (1) the total accumulated power of the validators that join the attack (in the order specified by the random arrival process) first exceeds the power threshold $t$, or (2) the attacking coalition size remains under the detectability threshold $1-\eta$, and hence the attackers remain undetected. We formally define these two conditions in~\Cref{def:Kpow,def:K-V-Pi} below. 

\begin{definition}[Power-only stopping time] \label{def:Kpow}
    \begin{equation*}
        K_p(t)\coloneqq\min\{j\geq1\mid S_j(\pi)\geq t\}.
    \end{equation*}
\end{definition}

\begin{definition}[Stopping time with coalition size requirement] \label{def:K-V-Pi}
    Let $k\leq n$ be the size of coalition of attackers. The stopping time with coalition size requirement is defined as follows.
    \begin{align*}
    K(t,k)\coloneqq&K_p(t)\vee k\\
    =&\min\{j\geq k\mid S_j(\pi)\geq t\}.
    \end{align*}
\end{definition}
\begin{remark}
    The detectability parameter $\eta$ in \cite{yeo} corresponds to $\frac{n-k}{n}$ in this article.
\end{remark}

\section{Whales and minnows} \label{sec:whale-minnow}
In this section, we study a particular distribution of voting power of practical interest.

We suppose there are only two types of validator nodes: whales $W$ with voting power $\alpha$ and minnows $M$ with voting power $\varepsilon$. 
Whales correspond to validators with large voting powers, and minnows correspond to validators with significantly smaller voting powers (compared to the whales).
The number of whales is $w$ and the number of minnows is $n-w$.

We assume that the number of whales is far less than that of minnows and the voting powers of whales are far greater than those of minnows, i.e.,
\begin{align*}
    &w=\mathcal{O}(1),\\
    &\frac{\alpha}{\varepsilon}=\mathcal{O}(n).
\end{align*}
Note that $\alpha/\varepsilon=\mathcal{O}(n)$ is where the model is interesting. Suppose for example that $\alpha/\varepsilon=\Theta(n^2)$. Then the collective voting power of minnows are $\mathcal{O}\left(\frac{1}{n}\right)=o(1)$: The minnows collectively hold near $0$ voting power in the large $n$ limit. Then being fair or unfair ceases to be a question.

We first note that in the absence of any attack, i.e., all validators act according to the honest strategy, the total expected utility for the whales is $\alpha w$, which we henceforth call the \emph{whale fair play value}.

\begin{definition}[Attacking whales] \label{def:attacking-whales}
    Define the number of whales appearing in the first $k$ arrivals as follows.
    \begin{equation*}
        N_k:=\left|\{1\leq i\leq k\mid \pi^{-1}(i)\in W\}\right|
    \end{equation*}
\end{definition}
Note that $N_k \sim \mathrm{Hypergeometric}(n,w,k)$, which comes from the assumption of the uniformly random arrival order of the validators (see~\Cref{def:pi}). Hence, the expected number of whales in the first $k$ arrivals is
\begin{equation*}
    \E[N_k]=\frac{kw}{n}.
\end{equation*}
And the variance is
\begin{equation*}
    \text{Var}[N_k]=\frac{kw(n-k)(n-w)}{n^2(n-1)}.
\end{equation*}
Using the number of whales appearing in the first $k$ arrivals $N_k$, we can now decompose the $k$-sum into two parts: the first part comprises of the total voting power held by the whales in $N_k$, and the second part comprises of the total voting power held by the remaining minnows.

\begin{equation}\label{eq:sum_decomp}
    S_k(N_k)=\alpha N_k+\varepsilon(k-N_k).
\end{equation}

Rearranging~\Cref{eq:sum_decomp}, we get that the following two events are identical. 

\begin{equation*}
    \{S_k\geq t\}=\left\{N_k\geq\left\lceil\frac{t-k\varepsilon}{\alpha-\varepsilon}\right\rceil\right\}.
\end{equation*}

Importantly, the above equation effectively equates the success of the attack (i.e., $\{S_k\geq t\}$ on the LHS) to the number of whales appearing in the first $k$ arrivals (i.e., $\left\{N_k\geq\left\lceil\frac{t-k\varepsilon}{\alpha-\varepsilon}\right\rceil\right\}$ on the RHS). 

\begin{definition} \label{def:n0}
    Define the minimum number of whales in first $k$ arrivals so that $S_k\geq t$ as follows.
    \begin{equation*}
        N_0\coloneqq\max\left\{0,\left\lceil\frac{t-k\varepsilon}{\alpha-\varepsilon}\right\rceil\right\}.
    \end{equation*}
\end{definition}

We now define the whale's gain as the expectation of the fraction of attacking whales' voting power over the attackers' voting power. This corresponds to the expected potential \emph{gain} in block rewards for the whales that successfully join the attack, assuming block rewards are split proportionately to power.

To motivate this definition we price the whales as a species, i.e., we track the reward of the $w$ whales collectively. Under the honest strategy each block reward is split in proportion to voting power, so every validator earns its own share in the long run and the whales earn their fair play value $\alpha w$. Under a successful attack only the coalition receives and divides the reward, so $\frac{1}{S_K}$ acts as a multiplier on each member's power. The coalition formed at $K$ holds $N_K$ whales of power $\alpha$ each, while the remaining $w-N_K$ whales are censored and earn nothing; fixing the arrival order, the strong law of large numbers shows that the whales' long-run average reward converges almost
surely to the random variable $\frac{\alpha N_K}{S_K}$, whose mean we take as the
whale's gain. Note that a successful attack need not be a success for the whales
as a species: those inside the coalition profit, but any whales left outside lose
everything---a loss already reflected in the numerator $\alpha N_K$.

\begin{definition}[Whale's gain] \label{def:whale's-gain}
    \begin{equation*}
        \Phi(k,t)\coloneqq\E\left[\frac{\alpha N_K}{S_K}\right]
    \end{equation*}
\end{definition}
It follows from the law of total probability that
\begin{align*}
    \Phi(k,t)&=\E\left[\frac{\alpha N_k}{S_k}\mathbf{1}\left\{S_k\geq t\right\}\right]+\E\left[\frac{\alpha N_{K_p}}{S_{K_p}}\mathbf{1}\{S_k<t\}\right]\\
    &\coloneqq\Phi_1(k,t)+\Phi_2(k,t).
\end{align*}

The following two lemmas decompose and bound each component of the whale's gain.~\Cref{lemma:phi-1} computes the whale's gain in the case where the attack succeeds.~\Cref{lemma:phi-2} upper bounds the whale's gain in the case of a failed attack.

\begin{lemma} \label{lemma:phi-1}
    \begin{equation*}
        \Phi_1(k,t)=\sum_{j= N_0}^w\Prob(N_k=j)\frac{\alpha j}{\alpha j+\varepsilon(k-j)}.
    \end{equation*}
\end{lemma}
\begin{proof}
    This follows from
    \begin{equation*}
        \E\left[\frac{\alpha N_k}{S_k}\mathbf{1}\left\{S_k\geq t\right\}\right]=\E\left[\frac{\alpha N_k}{\alpha N_k+\varepsilon(k-N_k)}\mathbf{1}\{N_k\geq N_0\}\right].
    \end{equation*}\qed
\end{proof}
\begin{lemma} \label{lemma:phi-2}
    \begin{equation*}
        \Phi_2(k,t)\leq\Prob(N_k<N_0)\frac{\alpha N_0}{t}.
    \end{equation*}    
\end{lemma}
\begin{proof}
    On $S_k<t$, the power-only stopping time is in effect. Hence,
    \begin{equation*}
        K_p\geq k+1
    \end{equation*}
    and
    \begin{equation*}
        S_{K_p-1}<t\leq S_{K_p}<t+\alpha.
    \end{equation*}
    It follows that
    \begin{align*}
        (\alpha-\varepsilon)N_{K_p}&=S_{K_p}-\varepsilon K_p\\
        &<t+\alpha-\varepsilon(k+1),
    \end{align*}
    where the first line is a rearrangement of  $S_{K_p}=\alpha N_{K_p}+\varepsilon\left(K_p-N_{K_p}\right)$.
    
    Dividing by $\alpha-\varepsilon$,
    \begin{equation*}
        N_{K_p}<\frac{t-k\varepsilon}{\alpha-\varepsilon}+1\leq N_0+1\implies N_{K_p}\leq N_0.
    \end{equation*}
    Hence,
    \begin{equation*}
        \Phi_2\leq\E\left[\frac{\alpha N_0}{t}\mathbf{1}\{S_k<t\}\right]=\Prob(N_k<N_0)\frac{N_0\alpha}{t}.
    \end{equation*}\qed
\end{proof}
We search for an upper bound of the whale's gain $\Phi$ independent of the number of attackers $k$. First, let $m\in\mathbb{N}$ be an upper bound on $N_0$ uniform in $k$. That is, $m \geq N_0$, where $N_0$ can depend on $k$, but $m$ does not depend on $k$.
We first show an upper bound of $\Phi$ using the hypergeometric distribution in~\Cref{proposition:tail-bound} below.

\begin{restatable}[]{proposition}{tailbound}
\label{proposition:tail-bound}
For $m\geq N_0$,
    \begin{equation*}
        \Phi(k,t)<\sum_{j=m}^w\Prob(N_k=j)\frac{\alpha j}{\alpha j+\varepsilon(k-j)}+\Prob(N_k<m)\frac{\alpha m}{t}\coloneqq\overline{\Phi}(k,t).
    \end{equation*}
    We call the RHS of the above the \emph{hypergeometric upper bound} of the whale's gain.
\end{restatable}

\begin{proof}
    By \Cref{def:n0}, for all integers $j\in[N_0,m]$, we have $S_k(j)>t$. Hence,
    \begin{align*}
        \sum_{j=N_0}^{m-1}\Prob(N_k=j)\frac{\alpha j}{\alpha j+\varepsilon(k-j)}&<\sum_{j=N_0}^{m-1}\Prob(N_k=j)\frac{\alpha m}{t}\\        &=\Prob(N_0\leq N_k<m)\frac{\alpha m}{t}.
    \end{align*}
    By \Cref{lemma:phi-1,lemma:phi-2},
    \begin{align*}
        \Phi(k,t)&\leq\sum_{j=N_0}^w\Prob(N_k=j)\frac{\alpha j}{\alpha j+\varepsilon(k-j)}+\Prob(N_k<N_0)\frac{\alpha N_0}{t}\\
        &=\left(\sum_{j=N_0}^{m-1}+\sum_{j=m}^w\right)\Prob(N_k=j)\frac{\alpha j}{\alpha j+\varepsilon(k-j)}+\Prob(N_k<N_0)\frac{\alpha N_0}{t}\\
        &<\sum_{j=m}^w\Prob(N_k=j)\frac{\alpha j}{\alpha j+\varepsilon(k-j)}+\Prob(N_k<m)\frac{\alpha m}{t}.
    \end{align*}\qed
\end{proof}

\Cref{proposition:tail-bound} requires $m>N_0$, which is equivalent to the following condition on $k$
\begin{equation}
   k_m\coloneqq\left\lceil m+\frac{t-\alpha m}{\varepsilon}\right\rceil\leq k.
\end{equation}
We define the following auxiliary quantity
\begin{equation*}
    \Lambda(m)\coloneqq\sum_{j=0}^{m-1}\binom{w}{j}
\end{equation*}
with the convention $\Lambda(0)=0$.

\begin{restatable}[]{lemma}{probbound}
\label{lemma:prob-bound}
For $w\leq k<n$ and $m<w+1$,
    \begin{equation*}
        \Prob(N_k<m)\leq\Lambda(m)\left(\frac{n-k}{n-w}\right)^{w-m+1}.
    \end{equation*}
\end{restatable}

\begin{proof}
    For $0\leq j\leq w$
    \begin{align*}
        \Prob(N_k=j)&=\binom{w}{j}\prod_{i=0}^{j-1}\frac{k-i}{n-i}\prod_{i=0}^{w-j-1}\frac{n-k-i}{n-j-i}\\
        &<\binom{w}{j}\prod_{i=0}^{w-j-1}\frac{n-k}{n-j-(w-j)}\\
        &=\binom{w}{j}\left(\frac{n-k}{n-w}\right)^{w-j}\\
        &<\binom{w}{j}\left(\frac{n-k}{n-w}\right)^{w-m+1}.
    \end{align*}
    Summing over $j$, the result is proved.\qed
\end{proof}

We now have the requisite ingredients to derive an upper bound of the upper bound of the whale's gain $\overline{\Phi}$, which is independent of the attacking coalition size $k$.
Recall that the whale fair play value is the expected return of the whales when all validators follow the honest strategy.
In the following proposition, we show that for $k$ within a certain interval, the upper bounded whale's gain $\overline{\Phi}$ is upper bounded by the whale fair play value $\alpha w$. 
This immediately implies, assuming the upper bound on the whale's gain is tight enough, the existence of attacking coalition size parameter regions whereby whales might not join the attack, as they get a larger expected reward in the case where everyone behaves honestly.

\begin{restatable}[Uniform bound of $\overline{\Phi}$]{proposition}{uniformbound}
\label{proposition:uniform-bound-for-hypergeometric}
$\overline{\Phi}(k,t)<\alpha w$ for all $k$ satisfying the following inequality
    \begin{equation} \label{eqn:algebraic-inequality}
        \frac{\alpha(\alpha-\varepsilon)\varepsilon k^2w(n-k)(n-w)}{n^2(n-1)}\geq\frac{\alpha m}{t}\Lambda\left(\frac{n-k}{n-w}\right)^{w-m+1}.
    \end{equation}
\end{restatable}

\begin{proof}
    Define
    \begin{equation*}
        g(j,k)\coloneqq\frac{\alpha j}{\alpha j+\varepsilon(k-j)}.
    \end{equation*}
    Then,
    \begin{equation*}
        \frac{\partial^2g}{\partial j^2}=\frac{-2\alpha(\alpha-\varepsilon)\varepsilon k}{[\alpha j+\varepsilon(k-j)]^3}\leq-2\alpha(\alpha-\varepsilon)\varepsilon k.
    \end{equation*}
    Hence, $g(j,k)$ is $2\alpha(\alpha-\varepsilon)\varepsilon k$-strongly concave in $j$.

    By Jensen's inequality for strongly concave functions,
    \begin{align*}
        g(\E[N_k],k)-\E[g(N_k,k)]&\geq\frac{2\alpha(\alpha-\varepsilon)\varepsilon k}{2}\text{Var}[N_k]\\
        &=\frac{\alpha(\alpha-\varepsilon)\varepsilon k^2w(n-k)(n-w)}{n^2(n-1)}\\
        &\coloneqq f(k).
    \end{align*}
    On the other hand,
    \begin{align*}
        &\E[g(N_k,k)]=\sum_{j=0}^w\Prob(N_k=j)g(j,k),\\
        &g(\E[N_k],k)=\frac{\alpha\E[N_k]}{\alpha\E[N_k]+\varepsilon(k-\E[N_k])}=\frac{\alpha w}{\alpha w+\varepsilon(n-w)}=\alpha w.
    \end{align*}
    Then, for $m\geq1$,
    \begin{align*}
        \overline{\Phi}(k,t)&=\E[g(N_k,k)]+\sum_{j=0}^{m-1}\Prob(N_k=j)\left[\frac{\alpha m}{t}-g(j,k)\right]\\
        &\leq\alpha w-f(k)+\sum_{j=0}^{m-1}\Prob(N_k=j)\left[\frac{\alpha m}{t}-g(j,k)\right]\\
        &\coloneqq\alpha w-f(k)+\Delta(k).
    \end{align*}
    Hence, $f(k)>\Delta(k)\implies\overline{\Phi}(k,t)<\alpha w.$
    
    By \Cref{lemma:prob-bound}
    \begin{equation*}
        \Delta(k)<\sum_{j=0}^{m-1}\Prob(N_k=j)\frac{\alpha m}{t}<\frac{\alpha m}{t}\Lambda\left(\frac{n-k}{n-w}\right)^{w-m+1}.
    \end{equation*}
    Hence, \Cref{eqn:algebraic-inequality} is a sufficient condition for $f(k)>\Delta(k)$.\qed
\end{proof}

We now extract from \Cref{eqn:algebraic-inequality} an explicit condition on $k$. 

\begin{restatable}[]{lemma}{uniqueroot}
\label{lemma:unique-root}
For $m<w$, the following equation in $k$ has a unique root $\hat{k}\in[0,n)$.
    \begin{equation*}
        \frac{\alpha(\alpha-\varepsilon)\varepsilon k^2w(n-k)(n-w)}{n^2(n-1)}=\frac{\alpha m}{t}\Lambda\left(\frac{n-k}{n-w}\right)^{w-m+1}.
    \end{equation*}
\end{restatable}

\begin{proof}
    If $m=0$, $k=0$ is the unique solution.

    Suppose that $m\geq1$. Define
    \begin{equation*}
        f(k)\coloneqq\frac{Ck^2}{(n-k)^{w-m}}-1,
    \end{equation*}
    where
    \begin{equation*}
        C\coloneqq\frac{(\alpha-\varepsilon)\varepsilon(n-w)^{w-m+2}wt}{m\Lambda n^2(n-1)}=\frac{(\alpha n-1)(1-\alpha w)(n-w)^{w-m}wt}{m\Lambda n^2(n-1)}.
    \end{equation*}
    Then any root of the original equation is a root of $f(k)=0$ and vice versa.
    
    Since $f(0)=-1$ and $\lim_{k\to n^-}f(k)=+\infty$, there exists a root in $(0,n)$.

    Since $m<w$, $f$ is increasing in $k$. Hence, the root is also unique.\qed
\end{proof}

By the monotonicity of $f(k)$, the condition in \Cref{proposition:uniform-bound-for-hypergeometric} is equivalent to $\hat{k}\leq k<n$.

\section{Analysis of the validity range}\label{sec:validity}
In the previous section, we showed in~\Cref{proposition:uniform-bound-for-hypergeometric} an upper bound on the upper bounded whale's gain. In this section, we determine the conditions under which a reduction on whales' profit is guaranteed, i.e., we show an upper bound on the actual whale's gain. 
Indeed, our main theorem in this section~\Cref{thm:validity-range} shows that there is certain parameters such that the whale gain under the attack is actually smaller than in the setting where there is no attack, implying that in such a whale-minnow model the rational censorship attack of~\cite{yeo} might not always be successful.

In order to determine the parameter settings whereby the attack is not guaranteed to be successful, we first pin down the range of $m$. To do so, we summarise the conditions that $m$ must satisfy.
\begin{enumerate}
    \item $m\geq N_0(k)$ from \Cref{proposition:tail-bound};
    \item $m<w+1$ from \Cref{lemma:prob-bound};
    \item $m<w$ from \Cref{lemma:unique-root}.
\end{enumerate}
From condition 1, we have $k\geq\left\lceil m+\frac{t-\alpha m}{\varepsilon}\right\rceil=\left\lceil m+\frac{(t-\alpha m)(n-w)}{1-\alpha w}\right\rceil$. 
Since the upper bound of $k$ has to be strictly smaller than $n$, this gives us the following lower bound on $m$: 
\begin{equation*}
    m+\frac{(t-\alpha m)(n-w)}{1-\alpha w}\leq n-1\iff m\geq\frac{t(n-w)-(1-\alpha w)(n-1)}{\alpha n-1}.
\end{equation*}
Together with condition $3$, we obtain the feasible set $\mathcal{F}$ for $m$ as follows:
\begin{equation} \label{eqn:m-range}
    \mathcal{F}\coloneqq\mathbb{N}\cap\left[\frac{t(n-w)-(1-\alpha w)(n-1)}{\alpha n-1},w-1\right].
\end{equation}

We now show that the feasible set of $m$ is non-trivial.

\begin{proposition}\label{prop:nontrivial}
    $\mathcal{F}$ is not empty.
\end{proposition}
\begin{proof}
    First,
    \begin{equation*}
        \frac{t(n-w)-(1-\alpha w)(n-1)}{\alpha n-1}\leq\frac{t-(1-\alpha w)}{\alpha}.
    \end{equation*}
    Then,
    \begin{equation*}
        \frac{t-(1-\alpha w)}{\alpha}\leq w-1\iff\alpha\leq 1-t,
    \end{equation*}
    which is the original assumption on $\alpha$.\qed
\end{proof}

~\Cref{prop:nontrivial} implies that we can always choose an $m$ s.t. $k_m\leq n-1$.

It remains to show that we can choose $\hat{k}(m,n)\leq n-1$. Since $\hat{k}$ is the unique solution of $f(k)=0$ and $f$ is increasing in $k$,
\begin{equation*}
    \hat{k}\leq n-1\iff f(n-1)\geq0\iff C(m,n)(n-1)^2\geq1.
\end{equation*}
Let
\begin{align*}
    G(m,n)&\coloneqq C(m,n)(n-1)^2\\
    &=\frac{(\alpha n-1)(1-\alpha w)(n-w)^{w-m}(n-1)wt}{m\Lambda(m)n^2}.
\end{align*}
\begin{lemma} \label{lemma:g-m-decreasing}
    $G(m,n)$ is decreasing in $m$ for $m\in\mathcal{F}$ and increasing in $n$ for $n\geq w+1$.
\end{lemma}
\begin{proof}
    Since $n-w\geq1$ and $w-m\geq1$, higher $m$ leads to lower exponent. Also, $m\Lambda(m)$ in the denominator increases in $m$. Hence, $G$ decreases in $m$.

    For monotonicity in $n$, take the logarithmic derivative.
    \begin{align*}
       \frac{1}{G}\frac{\partial G}{\partial n}&=\frac{\alpha}{\alpha n-1}+\frac{w-m}{n-w}+\frac{1}{n-1}-\frac{2}{n}\\
       &>\frac{1}{n-w}+\frac{1}{n-1}-\frac{2}{n}\\
       &\geq\frac{2}{n-1}-\frac{2}{n}>0.
    \end{align*}\qed
\end{proof}
By \Cref{lemma:g-m-decreasing}, a sufficient condition for $\hat{k}\leq n-1$ is
\begin{equation*}
    G(w-1,n)\geq1
\end{equation*}
for all $n$.

Finally, we obtain our main result. Recall that~\Cref{proposition:uniform-bound-for-hypergeometric} gives an upper bound for $\overline{\Phi}(k,t)$, which is itself an upper bound on the whale's gain. Our main theorem presents an upper bound for the whale's gain $\Phi(k,t)$ itself.

\begin{theorem} \label{thm:validity-range}
    There exists $n_0(\alpha,w,t)$ s.t. for all $n\geq n_0$, there exists $m\in\mathcal{F}(n,\alpha,w,t)$ s.t.
    \begin{equation*}
        \Phi(k,t)<\alpha w
    \end{equation*}
    for all $\max\{k_m,\hat{k}(m,n),w\}\leq k<n$.
\end{theorem}
\begin{proof}
    It remains to determine $n_0$ so that $\hat{k}(m,n)<n$.

    Since $G$ is increasing in $n$, $n_0$ can be chosen as the smallest integer $n$ s.t.
    \begin{equation*}
        G(w-1,n)\geq1,
    \end{equation*}
    i.e.
    \begin{equation*}
        \frac{(\alpha n-1)(1-\alpha w)(n-w)(n-1)wt}{(w-1)(2^w-w-1)n^2}\geq1.
    \end{equation*}
    Since
    \begin{equation*}
         (\alpha n-1)(n-w)(n-1)>\alpha n^3-(\alpha w+\alpha+1)n^2,
    \end{equation*}
    it suffices to require
    \begin{equation*}
        \frac{[\alpha n-(\alpha w+\alpha+1)](1-\alpha w)wt}{(w-1)(2^w-w-1)}\geq1,
    \end{equation*}
    i.e.
    \begin{equation*}
        n\geq w+1+\frac{1}{\alpha}+\frac{(w-1)(2^w-w-1)}{\alpha(1-\alpha w)wt}.
    \end{equation*}
    Hence, we can choose $n_0(\alpha,w,t)=\left\lceil w+1+\frac{1}{\alpha}+\frac{(w-1)(2^w-w-1)}{\alpha(1-\alpha w)wt}\right\rceil$.\qed
\end{proof}

\section{Numerical experiments} \label{sec:experiment}
In this section, we choose the threshold $t=\frac{1}{2}$ and estimate $\Phi$ using Monte-Carlo for two of the larger and more widely-used proof-of-stake networks: Ethereum and Solana (\Cref{tab:top10-shares}).
\begin{table}[htbp!]
\centering
    \begin{tabular}{r l r @{\hspace{2em}} l r}
           & \multicolumn{2}{c}{Ethereum \cite{hildobby}} & \multicolumn{2}{c}{Solana \cite{helius}}\\
        \hline
        1  & Lido               & $22.3\%$ & Galaxy                       & $3.23\%$\\
        2  & Binance            & $8.3\%$  & Helius                       & $3.22\%$\\
        3  & ether.fi           & $5.2\%$  & Coinbase 02                  & $3.03\%$\\
        4  & Coinbase           & $4.5\%$  & Figment                      & $2.57\%$\\
        5  & Kraken             & $3.8\%$  & Ledger by Figment            & $2.44\%$\\
        6  & Figment            & $3.7\%$  & P2P.org                      & $1.98\%$\\
        7  & Blockdaemon        & $2.7\%$  & Everstake                    & $1.67\%$\\
        8  & Everstake          & $1.8\%$  & Binance staking              & $1.55\%$\\
        9  & Kiln               & $1.7\%$  & \texttt{G9x1...j6TY} (anon)  & $1.43\%$\\
        10 & Upbit              & $1.3\%$  & \texttt{H74q...UAst} (anon)  & $1.39\%$
    \end{tabular}
    \caption{Top $10$ stake shares in Ethereum and Solana.}
\label{tab:top10-shares}
\end{table}

\paragraph{Datasets.}
We took data pertaining to the number of validators as well as their amounts of stakes from~\cite{dune}. We note that although the number of validators in Ethereum is around $\sim1$M, each validator has the same amount of stake of $32$ETH. Thus we use the stake share percentages (as presented in~\Cref{tab:top10-shares}) to group the validators into whales with each whale controlling some number of validators corresponding the the stake share percentages above. This gives us a total of 10040 validators in Ethereum. We repeat the same process with Solana to get 1397 validators.

\paragraph{Empirical estimation of stake distribution and whale count.}
To check whether our model of whales and minnows is realistic and reflected in real-life blockchain systems, we tried to fit three parametric families of distributions that are mainly used to model data with very large and small values: exponential, power law, and stretched exponential function\footnote{The density of the stretched exponential function is $f(x)\propto e^{-\left(\frac{x}{\lambda}\right)^\beta}$.} to both the Ethereum and Solana stake distribution datasets.
The goodness of fit result is presented in~\Cref{tab:fit-comparison} and depicted in~\Cref{fig:stake-distribution}.
From this, we also obtain a natural choice for the whale count $w$ parameter for Ethereum and Solana.
Both exponential and stretched exponential distributions have a ``characteristic length'' $\lambda$, which is a natural candidate for the whale count. 

\begin{table}[htbp!]
\centering
    \begin{tabular}{l c c}
                    & $R^2$ (Ethereum) & $R^2$ (Solana)\\
        \hline
        Exponential & $0.971$ & $0.947$\\
        Power law & $0.927$ & $0.892$\\
        Stretched exponential & $0.994$ & $0.984$
    \end{tabular}
    \caption{The goodness of fit ($R^2$) for three parametric distribution families. 
    The power law distribution has the worst fit out of the 3 families.
    The stretched exponential is a better fit compared to the exponential distribution. The reason is that in both Ethereum and Solana there is a noticeable amount of users that control large amount of stake (i.e., the whales), so the distribution of stake drops significantly after accounting for these whales. 
    However, minnows in the tail region have roughly the same amount of stake, so the distribution drops very slowly in this region. A pure exponential distribution cannot capture this mismatch.}
\label{tab:fit-comparison}
\end{table}

\begin{figure}[htbp!]
\centering
    \includegraphics[width=\textwidth]{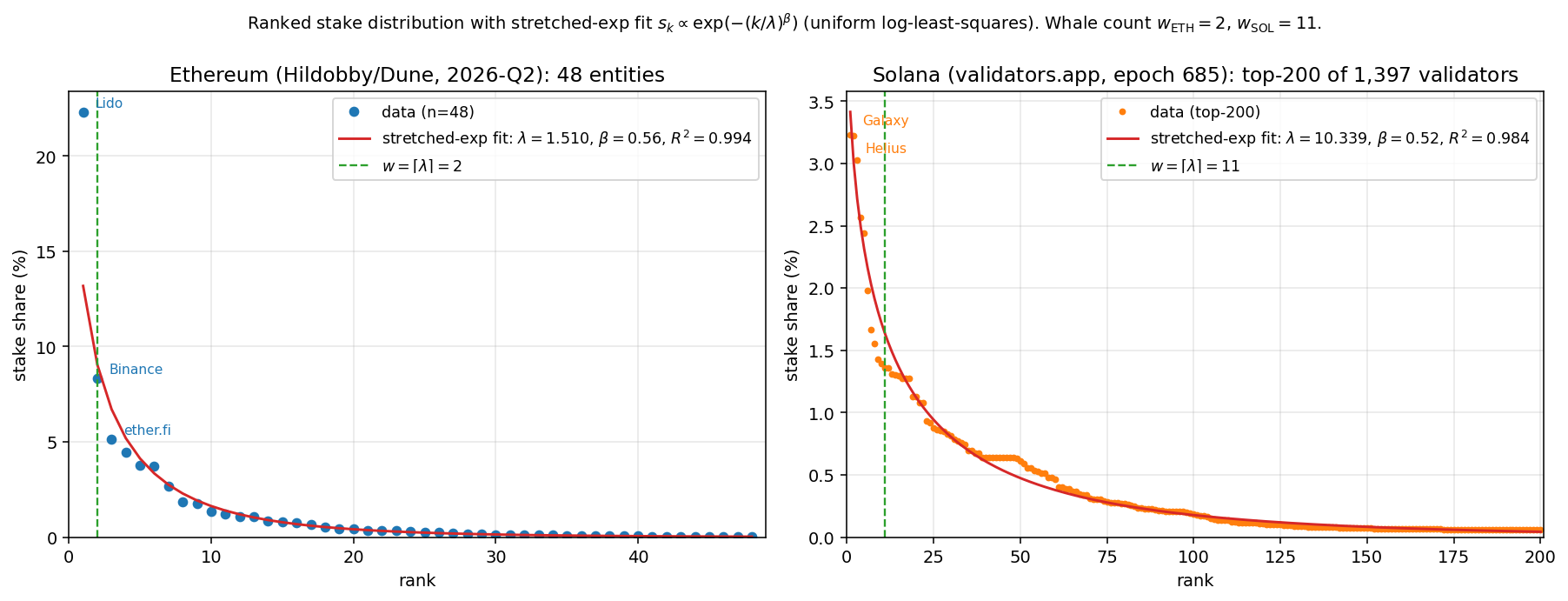}
    \caption{Fitting the stretched exponential distribution to Ethereum real world data in the plot on the left, and Solana in the plot on the right. The green dashed line: $w=\lceil\lambda\rceil$ is chosen as the whale count for each dataset.}
\label{fig:stake-distribution}
\end{figure}

\paragraph{Empirical estimation of whale's gain.}
Next, for each chain, we performed two experiments to empirically estimate the whale's gain depending on the number of attackers $k$: first, using the real whales' voting power (from real world data collected from~\cite{hildobby,helius}). 
Second, assuming all the whales have the same voting power.
We depict our results in~\Cref{fig:result}.
We note that all $4$ curves exhibit $3$ regimes: a relatively flat phase at a small attacking coalition size $k$, where the attack is favorable to whales; a transition phase where the attack is becoming detrimental to whales; a final upward phase that touches the whale fair play value. We stress that the final upward phase ending at the whale fair play value is not a coincidence: when $k=n$, every participant must be included in the attack for the attack being undetected. 
However, including all validators in the attack effectively nullifies the impact of the attack, rendering it to be exactly the case of where there is no attack, and every validator follows the honest strategy.

\begin{figure}[htbp!]
\centering
\includegraphics[width=\textwidth]{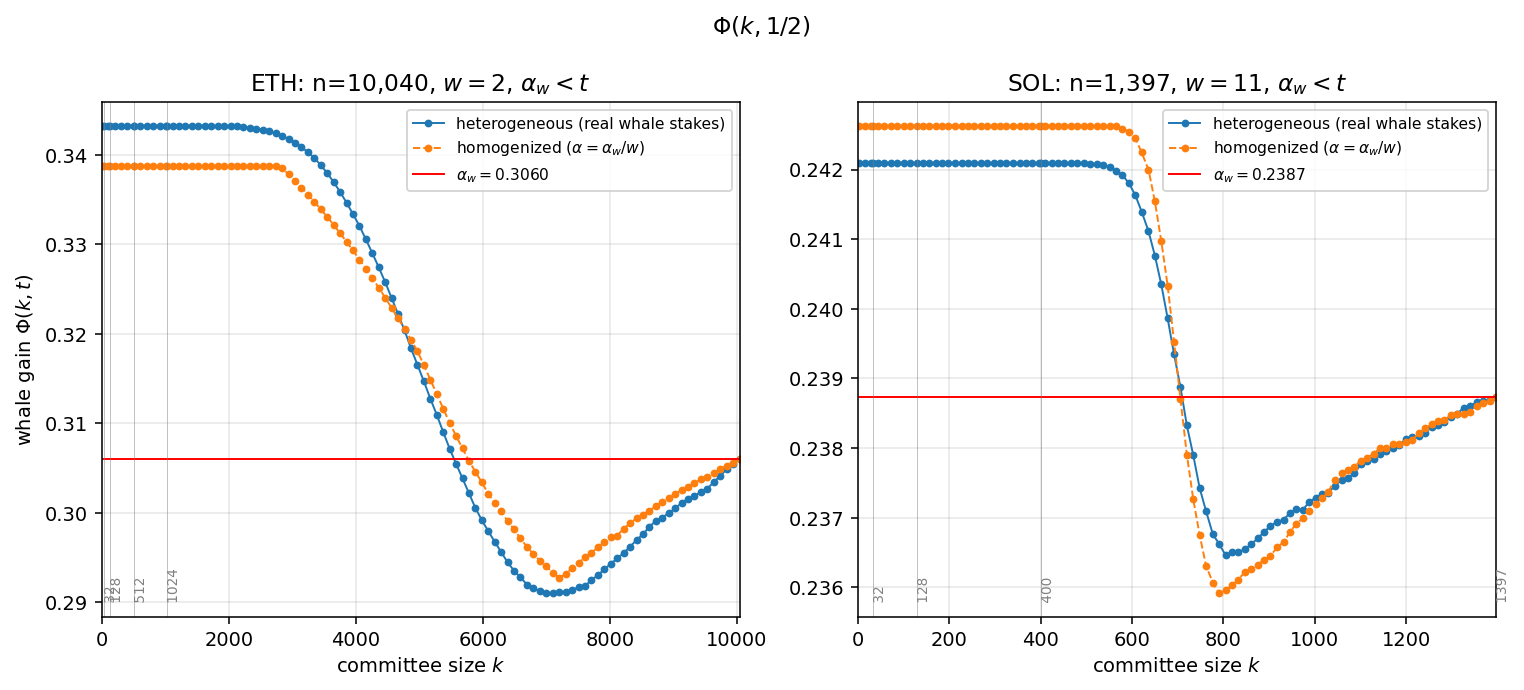}
\caption{$\Phi(k, t=1/2)$ vs. attack coalition size $k$. The plot on the left (resp. right) depicts the result for Ethereum (resp. Solana). The blue curve shows the estimation using real world whale stake numbers, and the orange curve shows the estimation using assuming all whales have the same stake. The red curve depicts the baseline whale fair play value.}
\label{fig:result}
\end{figure}

\paragraph{Verifying the validator lower bound.} 
We now verify if \Cref{thm:validity-range} applies to the setting where we assume all whales have the same voting power. Recall that~\Cref{thm:validity-range} gives a condition on the size of the validators to have a meaningful whale and minnow separation:
\begin{equation*}
    n\geq n_0=\left\lceil w+1+\frac{1}{\alpha}+\frac{(w-1)(2^w-w-1)}{\alpha(1-\alpha w)wt}\right\rceil.
\end{equation*}
For Ethereum, $n_0=19\ll n=10040$; for Solana, $n_0=224137\gg n=1397$. So \Cref{thm:validity-range} applies to Ethereum but says nothing about Solana. However, Monte Carlo simulation for Solana did show a reduction of whales' profit after $k>425$. 
This is because our theoretical bound of $n_0=\mathcal{O}(2^w)$ is very loose. Both \Cref{tab:top10-shares} and \Cref{fig:stake-distribution} show that Ethereum is more controlled by a few whales than Solana, which amounts to a lower number of whales $w$.

\paragraph{Attacking coalition size requirement.}
Lastly, we study the requirement on coalition size $\max\{k_m,\hat{k}(m,n),w\}$ in Ethereum given by \Cref{thm:validity-range}. The feasible set for $m$ is
\begin{equation*}
    \mathcal{F}(n,\alpha,w,t)=\{0,1\}. 
\end{equation*}
Since $m$ is a free parameter, we are free to choose it to minimize $k^*$. The result is
\begin{align*}
    k^{*}&\coloneqq\min_{m\in\mathcal{F}}\max\left\{w,k_m,\hat{k}(m,n)\right\}=7232,\\
    m^*&\coloneqq\arg\min_{m\in\mathcal{F}}\left\{\max\left\{w,k_m,\hat{k}(m,n)\right\}\right\}=0.
\end{align*}
Since $k^*<n$, the result is not vacuous. However, $k^*>k_{\text{exp}}=5781$ found by the Monte Carlo simulation, so the theoretical requirement is conservative.

\section{Discussion and implications} \label{sec:discussion}
We studied the censorship attack introduced in \cite{yeo}, where blockchain consensus participants can collude to exclude some other participants from earning any reward. We showed that in case that a finite amount of voting power is controlled by some handful amount of whales and the rest of voting power is uniformly distributed among a large number of minnows, if the detectability of this attack increase above certain threshold $k^*$, the attack will become detrimental to the whales (\Cref{thm:validity-range}). This will deter the whales from joining the attack in the first place. Without whales, it becomes difficult for minnows themselves to gather enough voting power to launch the attack.

\Cref{thm:validity-range} provides a conservative (way too large) estimate of the required amount of participants and the detectability level. We performed Monte Carlo simulation in \Cref{sec:experiment} to show that this reduction in whales' reward phenomenon occurs even when the condition of \Cref{thm:validity-range} is not met in Solana.

Our work initiates several interesting open questions and important directions of future work, which we highlight below.

\paragraph{Power centralisation.}
The threat of censorship is typically associated with centralisation of power. 
Perhaps somewhat counterintuitively, our work shows that the presence of whales might prevent censorship attacks under certain parameter regimes, and in particular with high detection costs and detection thresholds.
In contrast, the setting where there are no whales and only minnows is susceptible to the original censorship attack proposed by~\cite{yeo}.
This suggests that some degree of power centralisation can prevent certain types of censorship attacks, and it would be interesting to explore and analyse such censorship attacks under different voting power distributions.

\paragraph{Non-uniformly random arrival order.}
Recall that a key assumption underlying both the original attack and our model is the uniformly-random validator arrival order.
Another important and orthogonal direction of research is to consider the impact of a non-uniformly random arrival order on the attack success. 
In particular, the realistic setting where whales can influence the arrival order (whether through bribery or selection) and put themselves first in the attack. 
Such a non-uniformly random arrival process could influence both the success and detectability of the attack, and is an interesting direction of future research.

\bibliographystyle{splncs04}
\bibliography{references} 

\appendix
\section{Omitted proofs}\label{app:proofs}

\subsection{Proof of~\texorpdfstring{\Cref{proposition:uniform-bound-for-hypergeometric}}{Proposition 2}}\label{app:uniform}
\uniformbound*

\begin{proof}
    Define
    \begin{equation*}
        g(j,k)\coloneqq\frac{\alpha j}{\alpha j+\varepsilon(k-j)}.
    \end{equation*}
    Then,
    \begin{equation*}
        \frac{\partial^2g}{\partial j^2}=\frac{-2\alpha(\alpha-\varepsilon)\varepsilon k}{[\alpha j+\varepsilon(k-j)]^3}\leq-2\alpha(\alpha-\varepsilon)\varepsilon k.
    \end{equation*}
    Hence, $g(j,k)$ is $2\alpha(\alpha-\varepsilon)\varepsilon k$-strongly concave in $j$.

    By Jensen's inequality for strongly concave functions,
    \begin{align*}
        g(\E[N_k],k)-\E[g(N_k,k)]&\geq\frac{2\alpha(\alpha-\varepsilon)\varepsilon k}{2}\text{Var}[N_k]\\
        &=\frac{\alpha(\alpha-\varepsilon)\varepsilon k^2w(n-k)(n-w)}{n^2(n-1)}\\
        &\coloneqq f(k).
    \end{align*}
    On the other hand,
    \begin{align*}
        &\E[g(N_k,k)]=\sum_{j=0}^w\Prob(N_k=j)g(j,k),\\
        &g(\E[N_k],k)=\frac{\alpha\E[N_k]}{\alpha\E[N_k]+\varepsilon(k-\E[N_k])}=\frac{\alpha w}{\alpha w+\varepsilon(n-w)}=\alpha w.
    \end{align*}
    Then, for $m\geq1$,
    \begin{align*}
        \overline{\Phi}(k,t)&=\E[g(N_k,k)]+\sum_{j=0}^{m-1}\Prob(N_k=j)\left[\frac{\alpha m}{t}-g(j,k)\right]\\
        &\leq\alpha w-f(k)+\sum_{j=0}^{m-1}\Prob(N_k=j)\left[\frac{\alpha m}{t}-g(j,k)\right]\\
        &\coloneqq\alpha w-f(k)+\Delta(k).
    \end{align*}
    Hence, $f(k)>\Delta(k)\implies\overline{\Phi}(k,t)<\alpha w.$
    
    By \Cref{lemma:prob-bound}
    \begin{equation*}
        \Delta(k)<\sum_{j=0}^{m-1}\Prob(N_k=j)\frac{\alpha m}{t}<\frac{\alpha m}{t}\Lambda\left(\frac{n-k}{n-w}\right)^{w-m+1}.
    \end{equation*}
    Hence, \Cref{eqn:algebraic-inequality} is a sufficient condition for $f(k)>\Delta(k)$.\qed
\end{proof}

\end{document}